\documentclass{amsart}
\usepackage{amsmath,amssymb}

\newtheorem{theorem}{Theorem}[section]
\newtheorem{proposition}[theorem]{Proposition}

\newtheorem{corollary}[theorem]{Corollary}
\newtheorem{conjecture}[theorem]{Conjecture}
\newtheorem{definition}[theorem]{Definition}

\numberwithin{equation}{section}

\def\DJ{{\hbox{D\kern-.8em\raise.15ex\hbox{--}\kern.35em}}}
\def\DJo{$\;$\kern-.4em
    \hbox{D\kern-.8em\raise.15ex\hbox{--}\kern.35em okovi\'c}}

\def\NSERC{Supported in part by an NSERC Discovery Grant.}

\def\al{{\alpha}}
\def\be{{\beta}}
\def\ga{{\gamma}}
\def\de{{\delta}}
\def\sig{{\sigma}}

\def\la{{\lambda}}
\def\La{{\Lambda}}

\def\bC{{\mbox{\bf C}}}

\def\pH{{\mathcal{H}}}

\def\GL{{\rm GL}}
\def\Un{{\rm U}}
\def\diag{{\rm diag}}
\def\tr{{\rm tr\;}}

\def\me{{\psi_{\rm max}}}
\def\vek#1{|#1\rangle}
\def\kov#1{\langle#1|}

\begin{document}

\title{On two-distillable Werner states}

\author {Dragomir \v{Z}. \DJo}
\address{Department of Pure Mathematics and Institute for Quantum Computing, University of Waterloo,
Waterloo, Ontario, N2L 3G1, Canada}
\email{djokovic@uwaterloo.ca}
\thanks{\NSERC}


\date{}

\begin{abstract}
We consider bipartite mixed states $\rho$ in a $d \otimes d$ quantum system. We say that $\rho$ is PPT if its partial 
transpose $1 \otimes T (\rho)$ is positive semidefinite, 
and otherwise $\rho$ is NPT. 
The well-known Werner states are divided into three types: 
(a) the separable states (same as the PPT states), 
(b) the 1-distillable states (necessarily NPT), and 
(c) the NPT states which are not 1-distillable. 
We give several different formulations and provide further 
evidence for the validity of the conjecture that Werner states 
of type (c) are not 2-distillable. 
\end{abstract}

\maketitle 

\section{Introduction}

Let $\pH=\pH_A\otimes\pH_B$ be the Hilbert space for the quantum system 
consisting of two parties, A and B (Alice and Bob). We assume that the Hilbert spaces $\pH_A$ and $\pH_B$ 
have the same finite dimension, which we denote by $d$. A {\em product state} is a tensor product 
$\rho_A\otimes\rho_B$ of the states $\rho_A$ and $\rho_B$ of the first and second party, respectively. 
A bipartite state $\rho$ is {\em separable} if it can be written as a convex linear combination of product states. 
We say that a bipartite state is {\em entangled} if it is not separable. 
We say that $\rho$ is PPT if its partial transpose $\sig = 1 \otimes T(\rho)$, computed in some fixed 
orthonormal (o.n.) basis of $\pH_B$, is a positive semidefinite operator. 
Otherwise $\sig$ has a negative eigenvalue, and we say that $\rho$ is NPT.

It is more complicated to give the definition of distillability for 
bipartite states $\rho$. For that purpose we have to consider multiple copies of $\rho$. 
For $k$ copies the density matrix is the $k$th tensor power $\rho^{\otimes k}$ which acts on the 
Hilbert space $\pH^{\otimes k}$. We can identify $\pH^{\otimes k}$ with the tensor product of the Hilbert spaces 
$\pH_A^{\otimes k}$ and $\pH_B^{\otimes k}$. 
In this way we can view $\rho^{\otimes k}$ as a bipartite state. 
Thus any vector $\vek{\psi}\in\pH^{\otimes k}$ 
has its Schmidt decomposition and a well-defined Schmidt rank.

The definition of distillability given below is not the original one but it is the only one 
that we are going to use.
Replacing the original definition with this one was nontrivial, see \cite{MH2}.

\begin{definition} 
For a bipartite state $\rho$ acting on $\pH$ and an integer $k\ge1$, we say that $\rho$ is $k$-{\em distillable} 
if there exists a (non-normalized) pure state 
$\vek{\psi}\in\pH^{\otimes k}$ of Schmidt rank at most two such that
\begin{equation} \label{nejed}
\kov{\psi} \sig^{\otimes k} \vek{\psi} < 0, \quad 
\sig = 1 \otimes T(\rho).
\end{equation}
We say that $\rho$ is {\em distillable} if it is $k$-distillable 
for some integer $k\ge1$.
\end{definition}

If a bipartite state $\rho$ is separable then it is PPT, i.e.,
$\sig$ is positive semidefinite, and consequently $\rho$ is not
distillable. For the same reason, the entangled bipartate PPT states
are not distillable. Equivalently, every distillable bipartite state is necessarily NPT. 
It is not known whether the converse holds, i.e.,
whether every bipartite NPT state is distillable. 
However it is widely believed that the converse is false. 
In particular the following conjecture has been raised \cite{DiV}.

\begin{conjecture} \label{hip-1}
There exist bipartite NPT states which are not distillable.
\end{conjecture}

It is known \cite{JW} that for each integer $k\ge1$ there exist examples of bipartite states 
which are distillable but not $k$-distillable.

We fix an o.n. basis $\vek{i}$, $i=1,2,\ldots,d$ of $\pH_A$, and an o.n. basis of $\pH_B$ for which we use the same notation. The 
context will make clear which basis is used. After fixing these 
bases, we can define the {\em flip operator} $F:\pH \to \pH$ by 
$$ F=\sum_{i,j} \vek{i,j} \kov{j,i}. $$
The (non-normalized) Werner states on $\pH$ can be parametrized as
\begin{equation} \label{Werner-st}
\rho_W(t)=1-tF, \quad -1 \le t \le 1.
\end{equation}
Let $\vek{\me}\in\pH$ be the maximally entangled (pure) state given by
$$ \vek{\me}=\frac{1}{\sqrt{d}} \sum_i \vek{i,i}. $$
Its density matrix is the projector
$$ P=\frac{1}{d} \sum_{i,j} \vek{i,i} \kov{j,j}. $$
Since $dP$ is the partial transpose of $F$, the partial transpose
of $\rho_W(t)$ is 
$$ \sig_W(t)=1-tdP. $$

The following facts about the Werner states are well-known.
\begin{proposition} \label{Werner-prop}
The Werner states $\rho_W(t)$ are: 

(a) separable for $-1\le t\le 1/d$;

(b) $1$-distillable for $1/2<t\le1$;

(c) NPT but not 1-distillable for $1/d<t\le1/2$.
\end{proposition}

From now on, unless stated otherwise, we assume that $d\ge3$. 
(In section \ref{Matrica} we will consider briefly the case 
$d=2$.) The importance of Werner states for the distillability 
problem for bipartite states was first established in \cite{MH1}. 

\begin{proposition} \label{Horod-prop}
Conjecture \ref{hip-1} is equivalent to the assertion that some NPT Werner states $\rho_W(t)$ are not distillable. 
\end{proposition}

In fact the following stronger conjecture is believed to be true
\cite{LC,DiV,WD}.

\begin{conjecture} \label{hip-2}
None of the Werner states $\rho_W(t)$, $1/d<t\le1/2$, is distillable.
\end{conjecture}

In this paper we will consider a very weak version of it.
\begin{conjecture} \label{hip-3}
None of the Werner states $\rho_W(t)$, $1/d<t\le1/2$,
is $2$-distillable.
\end{conjecture}

For the $k$-distillability problem the following fact 
\cite[Lemma 4]{DiV} is important.
\begin{proposition} \label{DiV-prop}
If $\rho_W(1/2)$ is not $k$-distillable then none of the states 
$\rho_W(t)$, $1/d<t<1/2$, is $k$-distillable.
\end{proposition}

In view of this proposition, it suffices to prove Conjecture 
\ref{hip-3} for $t=1/2$ only. Extensive numerical evidence for
the validity of this conjecture in the case $d=3$ is presented in 
\cite{DiV,WD,ML} and \cite{RV}. In \cite{RV} it is also claimed that their numerical proof is rigorous. 
The case $d=4$ was analyzed in \cite{LP} but it remains open. 
For an alternative approach to the conjectures mentioned 
above see the very recent paper \cite{MRW}. 
The paper is organized as follows. 

In section \ref{Fi} we construct a hermitian biquadratic form $\Phi$ and show that 
Conjecture \ref{hip-3} is equivalent to $\Phi$ being positive semidefinite, $\Phi\ge0$.
The form $\Phi$ depends on $4d$ arbitrary vectors $x_i,y_i\in\pH_A$ 
and $u_i,v_i\in\pH_B$, $i=1,2,\ldots,d$. 

In section \ref{4-Matrice} we obtain the formula (\ref{Fi:4-mat}) which expresses $\Phi$ 
as a function of four matrices $X,Y,U,V$ of order $d$, 
where $X=[~x_1~x_2~\cdots~x_d~]$, etc. From this formula we deduce that $\Phi$ is 
invariant under an action of the product of two copies of the unitary group $\Un(d)$.

In section \ref{Matrica} we compute the matrix $H=H(X,Y)$ of 
$\Phi$ when the latter is viewed as a hermitian quadratic form 
in the $2d^2$ complex entries of $U$ and $V$. The entries of 
$X$ and $Y$ play the role of parameters. 
We restate Conjecture \ref{hip-3} as Conjecture \ref{hip-4} which asserts that $H\ge0$. 
After partitioning $H$ into four square blocks of order $d^2$, 
we show that the two diagonal blocks are positive definite 
matrices. We reduce the task of proving that $H\ge0$ to the case 
where $X$ is a diagonal matrix with positive diagonal entries. 
In the case $d=2$ we prove that $H\ge0$. 

In section \ref{Dijagonala} we prove that, for any $d$,  
$H(X,Y)\ge0$ when $X$ and $Y$ are diagonal matrices. 
We point out that $H(X,Y)$ is not diagonal even when both $X$ and $Y$ are. Since this is done for arbitrary $d$, and the proof is 
nontrivial, we view this fact as an important piece of evidence for the validity of Conjecture \ref{hip-3}. 

In section \ref{Uprosti} we prove that the inequality 
$H(X,Y)\ge0$ is equivalent to 
$H(\alpha X+\beta Y, \gamma X+\delta Y)\ge0$, 
where $\alpha\delta-\beta\gamma\ne0$. Hence, it suffices to prove 
the inequality $H(X,Y)\ge0$ when $X$ is singular. 

In section \ref{d=3} we consider the case $d=3$. To prove that 
$H(X,Y)\ge0$, we may assume that $X$ is singular. Hence $X$ has 
rank 1 or 2. We prove that $H(X,Y)\ge0$ when $X$ has rank 1.
We also show that the leading principal minor of $H$ of order 10 is a positive semidefinite polynomial. 

The superscripts $*$, T and $\dag$ denote the complex conjugation, the transposition and the adjoint, respectively. 
We denote by $M_m$ the algebra of complex matrices of 
order $m$, and by $I_m$ the identity matrix of $M_m$.

\section{The hermitian biquadratic form $\Phi$} \label{Fi}

Since we are going to use only one Werner state, the one for
$t=1/2$, we set 
$$ 
\rho_W=\rho_W(1/2)=1-F/2,\quad \sig_W=\sig_W(1/2)=1-dP/2. 
$$
Conjecture \ref{hip-3} is equivalent to the claim that the inequality

\begin{equation} \label{glavna}
\kov{\psi} \sig_W^{\otimes 2} \vek{\psi} \ge 0
\end{equation}
is valid for all $\vek{\psi}\in\pH^{\otimes 2}$ of Schmidt rank 
$\le2$. Such $\vek{\psi}$ can be written as 
$\vek{\psi}=\vek{\psi_1}+\vek{\psi_2}$, where 
$$ 
\vek{\psi_1}=\vek{x}\otimes\vek{u}, \quad
\vek{\psi_2}=\vek{y}\otimes\vek{v}. 
$$
Note that $\vek{x},\vek{y}\in\pH_A\otimes\pH_A$ while
$\vek{u},\vek{v}\in\pH_B\otimes\pH_B$.
We point out that we do not require $\vek{\psi_1}+\vek{\psi_2}$ to be
the Schmidt decomposition of $\vek{\psi}$, i.e., we do not require that 
$\langle x|y \rangle=\langle u|v \rangle=0$. The reason for this is to allow 
the vectors $\vek{x},\vek{y},\vek{u},\vek{v}$ to be completely arbitrary.

We can rewrite $\vek{\psi_1}$ and $\vek{\psi_2}$ as
$$ 
\vek{\psi_1}=\sum_{i,j}\vek{i,j,x_i,u_j}, \quad
\vek{\psi_2}=\sum_{i,j}\vek{i,j,y_i,v_j}. 
$$
The vectors $\vek{x_i}$ and $\vek{y_i}$ live in Alice's second copy of $\pH_A$, 
while $\vek{u_i}$ and $\vek{v_i}$ live in Bob's second copy of $\pH_B$. 
The summation is taken over all $i$ and $j$ in $\{1,2,\ldots,d\}$.
Consequently, we can view the LHS of (\ref{glavna}) as a function of $4d$ 
vectors $x_i,y_j,u_r,v_s$:
$$ 
\Phi(x_1,\ldots,x_d,y_1,\ldots,y_d,u_1,\ldots,u_d,v_1,\ldots,v_d)=
\kov{\psi}\sig_W^{\otimes 2}\vek{\psi}. 
$$
As
$$ 
\sig_W^{\otimes 2}=1-\frac{1}{2}(1\otimes dP+dP\otimes 1)
	+\frac{1}{4}dP\otimes dP, 
$$
we have
$$ 
\Phi=\Phi_1-\frac{1}{2}(\Phi_2+\Phi_3)+\frac{1}{4}\Phi_4, 
$$
where
\begin{eqnarray*}
\Phi_1 &=& \langle \psi|\psi\rangle,\\
\Phi_2 &=& \kov{\psi} 1 \otimes dP \vek{\psi},\\
\Phi_3 &=& \kov{\psi} dP \otimes 1 \vek{\psi},\\
\Phi_4 &=& \kov{\psi} dP \otimes dP \vek{\psi}.
\end{eqnarray*}
After the substitution $\vek{\psi}=\vek{\psi_1}+\vek{\psi_2}$, 
each of the $\Phi_k$ breaks up into four pieces. For instance, we have
\begin{eqnarray*}
\Phi_2 &=&
\sum_{i,j,r,s}\kov{i,j,x_i,u_j} 1 \otimes dP \vek{r,s,x_r,u_s} \\
&& +\sum_{i,j,r,s}\kov{i,j,x_i,u_j} 1 \otimes dP \vek{r,s,y_r,v_s} \\
&& +\sum_{i,j,r,s}\kov{i,j,y_i,v_j} 1 \otimes dP \vek{r,s,x_r,u_s} \\
&& +\sum_{i,j,r,s}\kov{i,j,y_i,v_j} 1 \otimes dP \vek{r,s,y_r,v_s}.
\end{eqnarray*}
We have computed each of the resulting 16 pieces. For instance the second piece, say $E$, 
in the above formula for $\Phi_2$ is computed as follows. We first observe that
$\kov{i,j,x_i,u_j} 1 \otimes dP \vek{r,s,y_r,v_s}=0$ if 
$r\ne i$ or $s\ne j$. Thus we have 

\begin{eqnarray*}
E &=& \sum_{i,j} \kov{x_i,u_j} dP \vek{y_i,v_j} \\
&=& \sum_{i,j,r,s} \kov{x_i,u_j}r,r\rangle\langle s,s\vek{y_i,v_j} \\
&=& \sum_{i,j} \left( \sum_r \langle x_i,u_j|r,r \rangle \cdot
\sum_s \langle s,s|y_i,v_j \rangle \right) \\
&=& \sum_{i,j} \langle x_i|u_j^*\rangle \langle y_i|v_j^*\rangle^*.
\end{eqnarray*}
The final formulas are:
\begin{eqnarray*}
\Phi_1 &=& \sum_i\|x_i\|^2\cdot\sum_j\|u_j\|^2
+\sum_i\langle x_i|y_i\rangle\cdot\sum_j\langle u_j|v_j\rangle \\
&& +\sum_i\langle y_i|x_i\rangle\cdot\sum_j\langle v_j|u_j\rangle
+\sum_i\|y_i\|^2\cdot\sum_j\|v_j\|^2, \\
\Phi_2 &=& \sum_{i,j} |\langle x_i|u_j^* \rangle|^2
+\sum_{i,j}\langle x_i|u_j^*\rangle \langle y_i|v_j^*\rangle^* \\
&& +\sum_{i,j}\langle y_i|v_j^*\rangle
\langle x_i|u_j^*\rangle^*+\sum_{i,j} |\langle y_i|v_j^*\rangle|^2, \\
\Phi_3 &=& \sum_{i,j}\langle x_i|x_j\rangle\langle u_i|u_j\rangle
+\sum_{i,j}\langle x_i|y_j\rangle\langle u_i|v_j\rangle \\
&& +\sum_{i,j}\langle y_j|x_i\rangle\langle v_j|u_i\rangle
+\sum_{i,j}\langle y_i|y_j\rangle\langle v_i|v_j\rangle, \\
\Phi_4 &=& \left| \sum_i \langle x_i|u_i^*\rangle \right|^2
+\sum_{i,j}\langle x_i|u_i^*\rangle \langle y_j|v_j^*\rangle^* \\
&& +\sum_{i,j} \langle x_i|u_i^*\rangle^* \langle y_j|v_j^*\rangle
+\left|\sum_j \langle y_j|v_j^*\rangle \right|^2.
\end{eqnarray*}

These formulas show that each $\Phi_k$, viewed as a function of the components of the $x_i$ and $y_j$, 
is a hermitian quadratic form. 
The same is true when we view them as functions of the components of the $u_i$ and $v_j$. 
Hence we shall refer to the $\Phi_k$ (and $\Phi$) as hermitian biquadratic forms. 
The next proposition follows immediately from  (\ref{glavna}) and the definition of the form $\Phi$.

\begin{proposition} \label{Ekv-Fi}
Conjecture \ref{hip-3} is equivalent to the assertion that 
$\Phi\ge0$.
\end{proposition}

\section{$\Phi$ as a function of four matrices}
\label{4-Matrice}

Let $X$ denote the $d\times d$ matrix whose successive columns are the vectors $x_1,\ldots,x_d$. 
Define similarly the matrices $Y,U$, and $V$. Let $M_d$ denote the space of complex matrices of order $d$. 
Define the inner product on $M_d$ by $\langle A|B\rangle=\tr(A^\dag B)$. 
For the corresponding norm we have $\|A\|^2=\tr(A^\dag A)$. 
The tensor product of matrices $A=[a_{ij}]$ and $B$ is defined as 
the block-matrix $A\otimes B=[a_{i,j}B]$. 

Now the formulas for $\Phi$ can be rewritten in terms of the matrices $X,Y,U,$ and $V$. We obtain that

\begin{eqnarray*}
\Phi_1(X,Y,U,V) &=& \|X\|^2\|U\|^2+\|Y\|^2\|V\|^2
+2\Re(\tr(X^\dag Y)\cdot\tr(U^\dag V)), \\
\Phi_2(X,Y,U,V) &=& \| X^T U+Y^T V \|^2, \\
\Phi_3(X,Y,U,V)&=& \tr \left( X^T X^* U^\dag U + X^T Y^* V^\dag U 
+ Y^T X^* U^\dag V + Y^T Y^* V^\dag V \right), \\
\Phi_4(X,Y,U,V)&=& |\tr(X^T U + Y^T V)|^2,
\end{eqnarray*}
where $\Re$ stands for ``the real part of''.

The first expression can be further simplified by using the standard Frobenius norm on the tensor product of matrices
$$ 
\Phi_1(X,Y,U,V) = \|X\otimes U+Y\otimes V\|^2. 
$$
The third expression also simplifies to
$$ 
\Phi_3(X,Y,U,V) = \| UX^T+VY^T \|^2. 
$$

Consequently, we have
\begin{eqnarray}
\Phi(X,Y,U,V) &=&  \label{Fi:4-mat} 
\|X\otimes U+Y\otimes V\|^2 \\ \notag
&&
-\frac{1}{2} \left( \| X^T U+Y^T V \|^2 +\|UX^T+VY^T\|^2 
\right) \\ \notag
&&  
+\frac{1}{4}
\left| \tr(X^T U+Y^T V) \right|^2.
\end{eqnarray}

The next proposition follows imediately from the above formulas. 
\begin{proposition} \label{Invar}
The identity 
\begin{equation} \label{Fi-ident}
\Phi(AXB,AYB,A^*UB^*,A^*VB^*) = \Phi(X,Y,U,V),
\end{equation}
holds true for arbitrary $X,Y,U,V\in M_d$ and $A,B\in\Un(d)$. 
\end{proposition}

\section{The matrix $H$ of the form $\Phi$}
\label{Matrica}

We shall consider the entries of $X$ and $Y$ as parameters and those of $U$ and $V$ as complex variables. 
Then $\Phi$ (and each $\Phi_k$) becomes a family of hermitian quadratic forms depending on the mentioned parameters. 
Let $H=H(X,Y)$ and $H_k=H_k(X,Y)$, $k=1,2,3,4$, be the 
matrices of the corresponding forms $\Phi$ and $\Phi_k$. 
These are hermitian matrices of order $2d^2$.

For any complex matrix $Z$ let $\tilde{Z}$ denote the column vector obtained by writing the columns of $Z$ 
one below the other starting with the first column, then the second, etc. Now we can express the relationship 
between the form $\Phi$ and its matrix $H$ by the formula
\begin{equation} \label{For-Mat}
\Phi(X,Y,U,V) = \left[ \begin{array}{c} \tilde{U} \\ \tilde{V} \end{array} \right]^\dag 
H(X,Y) \left[\begin{array}{c}\tilde{U} \\ \tilde{V} \end{array} \right].
\end{equation}

By using the formulas given in section \ref{Fi}, we obtain 
the following simple formulas

\begin{eqnarray} \label{eq:H1}
H_1 &=& \left[\begin{array}{cc}
\|X\|^2 & \tr(X^\dag Y) \\ \tr(Y^\dag X) & \|Y\|^2
\end{array}\right] \otimes I_{d^2}, \\ \label{eq:H2}
H_2 &=& \left[\begin{array}{cc}
X^\dag X & X^\dag Y \\ Y^\dag X & Y^\dag Y
\end{array}\right] \otimes I_d, \\ \label{eq:H3}
H_3 &=& \left[\begin{array}{ll}
I_d\otimes X^* X^T & I_d\otimes X^* Y^T \\ \label{eq:H4}
I_d\otimes Y^* X^T & I_d\otimes Y^* Y^T
\end{array}\right], \\
H_4 &=& \left[\begin{array}{c}
\tilde{X} \\ \tilde{Y} \end{array}\right]^* \cdot
\left[\begin{array}{c}
\tilde{X} \\ \tilde{Y} \end{array}\right]^T.
\end{eqnarray}
for the matrices $H_k$. Those for $H_1$ and $H_4$ are obvious. 
We omit the tedious but straightforward verification 
of the formulas for $H_2$ and $H_3$.

For $H$ we obtain the formula 
\begin{equation} \label{MatH}
H(X,Y)=H_1-\frac{1}{2}(H_2+H_3)+\frac{1}{4}H_4,
\end{equation}
and for its trace 
\begin{equation} \label{TragH}
\tr H(X,Y)=\left( d-\frac{1}{2} \right)^2 \left( \|X\|^2+\|Y\|^2 \right).
\end{equation}

In view of Proposition \ref{Ekv-Fi}, we can restate Conjecture 
\ref{hip-3} in the following equivalent form.

\begin{conjecture} \label{hip-4}
$H(X,Y)\ge0, \quad \forall X,Y\in M_d$.
\end{conjecture}

If $A\in \Un(d)$ and we replace $X$ and $Y$ with $AX$ and $AY$, respectively, 
then the $H_k$ undergo the transformation 
$Z\to (I_{2d}\otimes A^*)Z(I_{2d}\otimes A^T)$. In fact $H_1$ and 
$H_2$ remain fixed under this transformation. 

Similarly, if $B\in \Un(d)$ and we replace $X$ and $Y$ with $XB$ and $YB$, respectively, 
then the $H_k$ undergo the transformation 
$Z\to (I_2\otimes B^\dag\otimes I_d)Z(I_2\otimes B\otimes I_d)$. 
This time $H_1$ and $H_3$ remain fixed. 
In the case of $H_4$ one should use the formulas 

$$
\widetilde{AX}=(I_d\otimes A)\cdot \tilde{X}, \quad 
(\widetilde{YB})^T=(\tilde{Y})^T\cdot (B\otimes I_d), 
$$
which are not hard to verify.

Hence, the following proposition holds.

\begin{proposition} \label{Transf-H}
For $A,B\in \Un(d)$ we have
\begin{equation} \label{AB}
H(AXB,AYB)=(I_2\otimes B^\dag\otimes A^*)H(X,Y)
(I_2\otimes B\otimes A^T).
\end{equation}
\end{proposition}

Thanks to this proposition (or Proposition \ref{Invar}) we can simplify the task of proving Conjecture \ref{hip-4}. Indeed, 
it suffices to prove this conjecture when the matrix $X$ is 
diagonal and its diagonal entries are positive.

Let us partition $H(X,Y)$ into four square blocks of size $d^2$.
The first diagonal block depends only on $X$ and
the second one only on $Y$. 
By using (\ref{MatH}) and the formulas 
(\ref{eq:H1})-(\ref{eq:H4}) we obtain that

\begin{equation} \label{PartH}
H(X,Y)=\left[\begin{array}{cc} L(X) & L(X,Y) \\ 
L(X,Y)^\dag & L(Y) \end{array}\right],
\end{equation}
where 

\begin{equation} \label{eq:L}
L(X,Y)=\tr(X^\dag Y) I_{d^2} -\frac{1}{2} 
\left( X^\dag Y\otimes I_d + I_d\otimes X^* Y^T \right)
+\frac{1}{4}{\tilde{X}}^*{\tilde{Y}}^T,
\end{equation}
and $L(X):=L(X,X)$.

If $X$ and $Y$ are nonzero matrices, then the two diagonal blocks 
in (\ref{PartH}) are positive definite matrices. 
This is shown in the next proposition.

\begin{proposition} \label{GlavniMin}
If $X\ne0$ then $L(X)>0$.
\end{proposition}

\begin{proof} 
By Proposition \ref{Invar}, we may assume that 
$X=\diag(\la_1,\la_2,\ldots,\la_d)$
with $\la_1\ge\la_2\ge\cdots\ge\la_d\ge0$. 
Let $s=\|X\|^2=\sum\la_i^2$.
It follows from (\ref{eq:L}) that 
$L(X)=M+(1/4)\tilde{X}{\tilde{X}}^T$, 
where 
$$
M=\bigoplus_{i=1}^d \left( (s-\frac{\la_i^2}{2}) I_d
-\frac{1}{2}X^2 \right)
$$ 
is a diagonal matrix with the diagonal entries
$$ 
\mu_{ij}=s-(\la_i^2+\la_j^2)/2,\quad i,j=1,2,\ldots,d. 
$$
Since 
$$
\mu_{ij} \ge \mu_{1,1} =\la_2^2+\cdots+\la_d^2 \ge 0
$$ 
for all $i,j$, we have $L(X)\ge0$. As $X\ne0$, we have 
$\la_1>0$. If $\la_2>0$ then all  $\mu_{ij}>0$ and so 
$L(X)>0$. Otherwise $\la_i=0$ for $i>1$ and $L(X)$ is 
a diagonal matrix with positive diagonal entries.
Hence again $L(X)>0$.
\end{proof}

The matrix $H$ has order $2d^2$, but 
one can reduce the proof of Conjecture \ref{hip-4} to matrices of order $d^2$. 
This does not come for free since the smaller matrix will have a more complicated structure.
Recall that we may assume that $X$ is a diagonal matrix 
with positive diagonal entries. 
For simplicity we set $A=L(X)$, $B=L(X,Y)$ and $C=L(Y)$t in (\ref{PartH}).
Since $A>0$, it suffices to show that 
$S:=C-B^\dag A^{-1} B \ge 0$, 
see e.g. \cite[Proposition 8.2.3]{DB}. 
(As $X$ is diagonal, one can easily compute $A^{-1}$.)
Proving that $S\ge0$ may be somewhat easier than proving that $H\ge0$.  
We shall use this simplification to handle the case $d=2$ below.

Recall that $d\ge3$ by the assumption made earlier, 
but Conjecture \ref{hip-4} also makes sense for $d=1$ and $d=2$. 
However in these two cases the determinant of $H(X,Y)$ is identically 0. 
For $d=1$ we have $H_1=H_2=H_3=H_4$ and the conjecture is 
obviously valid. It is also valid for $d=2$. 

\begin{proposition} \label{Stav:d=2}
Conjecture \ref{hip-4} is true for $d=2$. 
\end{proposition}

\begin{proof}
We may assume that 
$X=\left[ \begin{array}{cc} a&0\\0&b \end{array} \right]$ with $a,b>0$. 
Let 
$Y=\left[ \begin{array}{cc} u_1&v_1\\u_2&v_2 \end{array} 
\right]$, and let us partition $H$ as in (\ref{PartH})
and set again $A=L(X)$, $B=L(X,Y)$ and $C=L(Y)$.
Let $t^4-c_1 t^3+c_2 t^2-c_3 t +c_4$ be the characteristic 
polynomial of $S:=C-B^\dag A^{-1} B$. 
A computation shows that $c_4=0$. Set 

\begin{eqnarray*}
p &=& a^2+b^2, \\
q &=& a^4+4a^2b^2+b^4, \\
r &=& p(|u_2|^2+|v_1|^2 )+|av_2-bu_1|^2.
\end{eqnarray*} 

After some tedious computations, we found the 
following formulas for the $c_i$:
\begin{eqnarray*}
2pqc_1 &=& 4(p^2+a^2b^2)|av_2-bu_1|^2
        +p(2a^2b^2+3q)(|u_2|^2+|v_1|^2), \\
4p^2qc_2 &=& q|av_2-bu_1|^4 \\
&& +p(7a^4+22a^2b^2+7b^4)(|u_2|^2+|v_1|^2)|av_2-bu_1|^2 \\
&& +2p^2 \left( (q+3a^2b^2) \left(|u_2|^2+|v_1|^2 \right)^2 
         +2(a^4+a^2b^2+b^4) \left| u_2v_1 \right|^2 \right. \\
&& \quad\quad\quad \left. 
+2 \left| abu_2v_1 + (av_2-bu_1)^2 \right|^2 \right), \\
4pqc_3 &=& 
r \left( \left| 2abu_2v_1 + (av_2-bu_1)^2 \right|^2 
+2a^2b^2 (|u_2|^4+|v_1|^4) \right. \\
&& \left. \quad +p(|u_2|^2+|v_1|^2)|av_2-bu_1|^2 
+2p^2 \left| u_2 v_1 \right|^2 \right ). \\
\end{eqnarray*}

Since $p,q,r>0$, we conclude that all coefficients $c_i\ge0$. 
Hence, $S\ge0$ (see e.g. \cite[Proposition 8.2.6]{DB}).
\end{proof} 

We shall consider the case $d=3$ in section \ref{d=3}.

\section{The diagonal case} \label{Dijagonala}

We say that a matrix pair $(X,Y)$ is {\em generic} if the 
matrices $X$ and $Y$ are linearly independent and some linear 
combination of them is nonsingular. 

In this section we prove that $H(X,Y)\ge0$ when both $X$ and $Y$ are diagonal matrices while $d$ is arbitrary. This appears to be 
a trivial case, but it is not so as $H(X,Y)$ is not diagonal 
even if $X$ and $Y$ are. We prove a slightly stronger result.

\begin{theorem} \label{DijSlucaj}
If $(X,Y)$ is a generic pair of diagonal matrices, then 
$H(X,Y)>0$.
\end{theorem}

\begin{proof}
We denote the diagonal entries of $X$ and $Y$ by $\la_1,\ldots,\la_d$ and $\mu_1,\ldots,\mu_d$ respectively. 
The hypothesis implies that $\la_k\ne0$ or $\mu_k\ne0$ for each $k$. 
After replacing $H$ with $\Pi H \Pi^T$ where $\Pi$ is a sutable permutation matrix, $H$ becomes direct sum 
of $d^2-d$ blocks of order 2 and an additional block of order $2d$. 
It suffices to show that each of these blocks is positive definite.

The blocks of order 2 are indexed by the integers $p=(i-1)d+j$, where $i,j\in\{1,2,\ldots,d\}$ and $j\ne i$. 
For such index $p$, the corresponding block of order 2 is the principal submatrix $H(p)$ of the original 
matrix $H$ corresponding to indices $p$ and $p+d^2$. Explicitly we have

$$ 
H(p)=\sum_{k=1}^d c_k \left[\begin{array}{cc}
|\la_k|^2 & \la_k^*\mu_k \\ \la_k\mu_k^* & |\mu_k|^2
\end{array} \right], 
$$
where $c_k=1$ for $k\ne i,j$ and $c_i=c_j=1/2$. Each matrix on the RHS is positive semidefinite of rank 1. 
If $H(p)$ is singular, then all of these matrices must be singular and must have the same kernel. 
This contradicts the linear independence of $X$ and $Y$. Hence $H(p)$ must be positive definite.

It remains to consider the block $B$ of size $2d$, i.e., the principal submatrix of $H$ 
corresponding to the indices $(i-1)d+i$ and $(i-1)d+i+d^2$ for $1\le i\le d$. We have 
$B=B_1-(B_2+B_3)/2+B_4/4$ where $B_k$ denotes the corresponding principal submatrix of $H_k$. 
Let us first consider the matrix $B'=B_1-(B_2+B_3)/2$. 
After a suitable simultaneous permutation of rows and columns, 
$B'$ breaks up into the direct sum of $d$ blocks
$G(i)$ of order 2, where $i\in\{1,2,\ldots,d\}$. Explicitly we have

$$ 
G(i)=\sum_{k\ne i} \left[\begin{array}{cc}
|\la_k|^2 & \la_k^*\mu_k \\ \la_k\mu_k^* & |\mu_k|^2 
\end{array} \right]. 
$$

Each $G(i)$ is positive semidefinite of rank 1 or 2. 
Thus in the decomposition $B=B'+B_4/4$ we have $B'\ge0$ and $B_4\ge0$.
If all $G(i)>0$ then $B'>0$, and so $B>0$.

It remains to consider the case where some $G(i)$, say $G(1)$, is singular. 
By Cauchy-Schwarz inequality, the vectors $(\la_2,\la_3,\ldots,\la_d)$ and 
$(\mu_2,\mu_3,\ldots,\mu_d)$ are linearly dependent. 
It follows that all other $G(i)$ must be positive definite. 
Consequently, the nullspace of $B'$ is 1-dimensional and is spanned by the 
column vector having all components 0 except the first which is 
$-\mu_2$ and $(d+1)$th which is $\la_2$. 
This vector is not killed by $B_4$, because $\la_1 \mu_2-\la_2 \mu_1\ne0$. 
Hence, we conclude that $B>0$.
\end{proof}

\begin{corollary} \label{cor:dijag}
Conjecture \ref{hip-4} is valid when $X$ and $Y$ are diagonal matrices.
\end{corollary}

\begin{proof}
This follows from the theorem because any pair of diagonal matrices 
can be approximated by a generic pair of diagonal matrices.
\end{proof}

\section{Reduction to the singular case}
\label{Uprosti}

Let us show that $H(X,Y)$ satisfies yet another identity. Let 
\begin{equation} \label{Mat-La}
\La=\left[\begin{array}{cc}\al&\be\\ \ga&\de\end{array}\right]
\in\GL_2(\bC)
\end{equation}
and
\begin{equation} \label{TrInv}
(\La^T)^{-1}=\left[\begin{array}{cc}\al'&\be'\\ \ga'&\de'\end{array}\right].
\end{equation}
By using
$$ 
\left[\begin{array}{cc}\al&\ga\\ \be&\de\end{array}\right]
\left[\begin{array}{cc}\al'&\be'\\ \ga'&\de'\end{array}\right]=
\left[\begin{array}{cc}1&0\\0&1\end{array}\right], 
$$ 
we deduce that
\begin{eqnarray*}
&& (\al X + \be Y) \otimes (\al' U + \be' V)+
   (\ga X + \de Y) \otimes (\ga' U + \de' V)=
X\otimes U + Y\otimes V, \\
&& (\al X + \be Y)^T (\al' U + \be' V)+
   (\ga X + \de Y)^T (\ga' U + \de' V)=
X^T U + Y^T V, \\
&&  (\al' U + \be' V) (\al X + \be Y)^T +
    (\ga' U + \de' V) (\ga X + \de Y)^T =
U X^T + V Y^T.
\end{eqnarray*}
Consequently, (\ref{Fi:4-mat}) implies that
$$ \Phi(\al X+\be Y,\ga X+\de Y,\al' U+\be' V,\ga' U+\de' V)=
\Phi(X,Y,U,V). $$
By using (\ref{For-Mat}) and the formula
$$ \left[\begin{array}{c} \al'\tilde{U}+\be'\tilde{V} \\ 
\ga'\tilde{U}+\de'\tilde{V} \end{array}\right]=
\left( (\La^T)^{-1}\otimes I_{d^2}\right)
\left[\begin{array}{c}\tilde{U} \\ \tilde{V}\end{array}\right], $$
we obtain the new identity
\begin{equation} \label{H-La}
H(\al X+\be Y,\ga X+\de Y)=(\La^*\otimes I_{d^2})H(X,Y)
(\La^T\otimes I_{d^2}).
\end{equation}

It suffices to prove the inequality $H(X,Y)\ge0$ for 
generic pairs $(X,Y)$ only. If $(X,Y)$ is generic, we can 
choose $\La\in\GL_2(\bC)$ such that $\al X + \be Y$ is 
a singular matrix. Thus the identity (\ref{H-La}) shows 
that it suffices to prove $H(X,Y)\ge0$ when $X$ is 
singular and $Y$ is invertible.

Yet another conjecture, which is simpler and 
stronger than Conjecture \ref{hip-4}, may be of interest. 
Let us introduce the real valued polynomial $D(X,Y)=\det H(X,Y)$. 
By taking the determinants in (\ref{AB}) we obtain that 
\begin{equation} \label{DHAB}
D(AXB,AYB)=D(X,Y),\quad \forall A,B\in\Un(d).
\end{equation}
From  (\ref{H-La}) we deduce that

\begin{equation} \label{DH-GL}
D(\al X+\be Y,\ga X+\de Y) = \left| \al\de-\be\ga \right|^{2d^2}D(X,Y)
\end{equation}
is valid when $\La$ is invertible. 
Since both sides are polynomials,
this identity must be valid for arbitrary $\La$.

Note that $D(X,0)=0$ for all matrices $X$. 
More generally, we claim that $D(X,Y)=0$ if $X$ and $Y$ 
are linearly dependent. Indeed, it suffices to choose a matrix $\La$ as in 
(\ref{Mat-La}) such that 
$\ga X+\de Y=0$ and apply  (\ref{DH-GL}). 
The converse of this claim is false but we conjecture that it is true in a weaker form.

\begin{conjecture} \label{hip-5}
If $d\ge3$ then $D(X,Y)\ne0$ for generic $(X,Y)$. 
\end{conjecture}

Theorem \ref{DijSlucaj} shows that this conjecture is true when the matrices $X$ and $Y$ are diagonal. 
As this conjecture deals with only one polynomial and has no positivity conditions whatsoever, 
it should be much easier to prove (or disprove).

\begin{proposition} \label{Determinanta}
Conjecture \ref{hip-4} is a consequence of Conjecture \ref{hip-5}.
\end{proposition}

\begin{proof}
Let $X_1$ and $Y_1$ be any matrices in $M_d$. 
We have to show that $H(X_1,Y_1)$ is positive semidefinite. 
Clearly it suffices to prove this when the pair $(X_1,Y_1)$ is generic. 
Let $(X_0,Y_0)$ be a generic pair of diagonal matrices. 
Then $H(X_0,Y_0)$ is positive definite by Theorem \ref{DijSlucaj}. 
Consequently, $D(X_0,Y_0)>0$ and all eigenvalues of $H(X_0,Y_0)$ are positive. 
We can join the pairs $(X_0,Y_0)$ and $(X_1,Y_1)$ by a continuous path $(X_t,Y_t)$, 
$0\le t\le 1$, such that $(X_t,Y_t)$ is generic for each $t$. 
By Conjecture \ref{hip-5}, $D(X_t,Y_t)\ne0$ for all $t$. 
Hence $H(X_t,Y_t)$ has no zero eigenvalues. Since the eigenvalues of 
$H(X_t,Y_t)$ are continuous functions of $t$ and they are all positive for $t=0$, 
they must all remain positive for all values of $t$. 
In particular this is true for $t=1$. We thus conclude that $H(X_1,Y_1)$ is positive definite.
\end{proof}

\section{The case $d=3$}
\label{d=3}

In this section we consider only the case $d=3$. As mentioned earlier, in order to prove that $H(X,Y)\ge0$ it suffices 
to do that in the case when $X$ is singular. So, the rank of 
$X$ is 1 or 2. We shall prove the inequality in the case 
when this rank is 1. 

\begin{proposition} \label{Rang=1}
If $X,Y\in M_3$ and some linear combination of $X$ and $Y$ has rank one, then $H(X,Y)\ge0$.
\end{proposition}

\begin{proof}
We may assume that $X$ and $Y$ are linearly independent and 
that $X$ has rank one. Since we can multiply $X$ by a nonzero 
scalar, by applying Proposition \ref{Transf-H} we may assume that 
$$
X=\left[ \begin{array}{ccc} 1&0&0\\0&0&0\\0&0&0 \end{array}
\right].
$$
By applying the same proposition, we may also assume that 
$$
Y=\left[ \begin{array}{ccc} 
a&u&v\\
x&b&0\\
y&0&c
 \end{array}
\right], 
$$
where $b,c,u,v\ge0$.

We partition the matrix $H=H(X,Y)$ as in (\ref{PartH}) and 
set $A=L(X)$, $B=L(X,Y)$, $C=L(Y)$. As explained in 
section \ref{Fi:4-mat}, it suffices to show that the matrix 
$S:=C-B^\dag A^{-1} B$ is positive semidefinite. Let 
$$
p(t)=\sum_{k=0}^9 (-1)^k c_k t^{9-k},\quad c_0=1,
$$
be the characteristic polynomial of $S$. The $c_k$ are 
polynomials in the real variables $b,c,u,v$ and the 
complex variables $x,y$ and their conjugates $x^*,y^*$. 
(The variable $a$ does not occur.)

Set $c_k=p_k/d_k$ where $d_k=2^k$ for $k<9$ and 
$d_9=d_8=256$. Then the $p_k$ are polynomials with integer 
coefficients. All these computations were performed by 
using Maple since the $p_k$ may have several thousand terms. 
We claim that the polynomials $p_k$ are positive 
semidefinite, i.e., they have nonnegative values for all 
real $b,c,u,v$ and all complex $x,y$. 
The inequality $H(X,Y)\ge0$ is a consequence of this claim.

To prove our claim, we construct positive semidefinite 
polynomials $q_k$, $k\in\{1,2,\ldots,9\}$, such that the 
difference $p_k-q_k|bux-cvy|^2$ is also a positive semidefinite polynomial. 
We have $q_1=q_2=0$. The other $q_k$ are given in the appendix. 
The $q_k$ are obiously positive semidefinite. The proof that the 
differences  $p_k-q_k|bux-cvy|^2$ are positive semidefinite 
requires the use of Maple (or some other software for symbolic 
algebraic computations). We just expand $p_k-q_k|bux-cvy|^2$ 
and check that all coefficients are nonnegative integers and 
all monomials that occur in the expansion are hermitian squares. 
For instance, we have 
\begin{eqnarray*}
p_1 &=& 5(u^2+v^2+|x|^2+|y|^2)+6(b^2+c^2), \\
p_2 &=& 41\left( (u^2+v^2)^2 +(|x|^2+|y|^2)^2 ) \right) \\
&& +62(b^2+c^2)^2 +6b^2c^2 \\
&& +91(u^2+v^2)(|x|^2+|y|^2) \\
&& +102 \left( b^2(u^2+|x|^2) +c^2(v^2+|y|^2) \right) \\
&& +108 \left( b^2(v^2+|y|^2) +c^2(u^2+|x|^2) \right). 
\end{eqnarray*}
\end{proof}

As an aside, we mention that in the case when
$$
X=\left[ \begin{array}{ccc} a&0&0\\0&b&0\\0&0&c \end{array}
\right], \quad 
Y=\left[ \begin{array}{ccc} 
u_1&v_1&w_1\\
u_2&v_2&w_2\\
u_3&v_3&w_3
 \end{array}
\right], 
$$
where $a,b,c>0$ and $u_i,v_i,w_i\in\bC$, the leading principal 
minor $\mu_{10}$ of $H$ of order 10 is a positive semidefinite 
polynomial. This follows from the following explicit expression 
for $\mu_{10}$ as a sum of squares of real polynomials: 

\begin{eqnarray*}
\mu_{10} &=& \frac{1}{512} (2a^2+b^2+c^2)^2 (a^2+2b^2+c^2) 
(a^2+b^2+2c^2) \cdot p,
\end{eqnarray*}
where
\begin{eqnarray*}
p &=& 2(a^2+2b^2+c^2)(a^2+b^2+2c^2) \cdot \\
&& \quad \left(
 (4a^4+b^4+c^4+5a^2(b^2+c^2)+4b^2c^2)|bw_3-cv_2|^2 \right. \\
&& \left. \quad\quad +(a^2+b^2)(a^2+c^2)
( |cu_1-aw_3|^2 +|av_2-bu_1|^2 ) \right) \\
&& + \left( a^6+b^6+c^6+ 11a^2b^2c^2 + 5( a^4(b^2+c^2) +b^4(a^2+c^2) +c^4(a^2+b^2) ) \right) \cdot q
\end{eqnarray*}
and
\begin{eqnarray*}
q &=& 2(a^2+2b^2+c^2)(a^2+b^2+2c^2)(|v_3|^2+|w_2|^2) \\
&&     +(a^2+2b^2)(a^2+b^2+2c^2)(|u_3|^2+|w_1|^2)  \\
&&     +(a^2+2c^2)(a^2+2b^2+c^2)(|u_2|^2+|v_1|^2).
\end{eqnarray*}

Note that the equality $\mu_{10}=0$ implies that $Y$ is a scalar multiple of $X$.

\section{Appendix}
\label{Dodatak}

We list here the polynomials $q_k$, $k>2$, used in section 
\ref{d=3}.

\begin{eqnarray*}
q_3 &=& 2, \\
q_4 &=& 11(u^2+v^2+|x|^2+|y|^2)+14(b^2+c^2), \\
q_5 &=& 22(u^4+v^4+|x|^4+|y|^4) +38(b^4+c^4) \\
&& +44(u^2+v^2+|x|^2+|y|^2) +52(u^2+v^2)(|x|^2+|y|^2) \\
&& +59 \left( b^2(u^2+|x|^2) +c^2(v^2+|y|^2) \right) \\
&& +65 \left( b^2(v^2+|y|^2) +c^2(u^2+|x|^2) \right)
+86b^2c^2, \\
q_6 &=& 296 b^2 c^2 ( u^2 + v^2 + |x|^2 + |y|^2 ) \\
&& +254( b^2 v^2 |y|^2 + c^2 u^2 |x|^2 ) \\
&& +225 ( b^2 + c^2 )( u^2 |y|^2 + v^2 |x|^2 ) \\
&& +202 ( b^2 u^2 |x|^2 + c^2 v^2 |y|^2 ) 
+198 b^2 c^2 ( b^2 + c^2 ) \\
&& +192 ( b^2 + c^2 )( u^2 v^2 + |x|^2 |y|^2 ) \\
&& +168 ( u^2 |x|^2 ( v^2 + |y|^2 ) 
+ v^2 |y|^2 ( u^2 + |x|^2 ) ) \\
&& +141 ( b^4 ( v^2 + |y|^2 ) + c^4 ( u^2 + |x|^2 ) ) \\
&& +116 ( b^4 ( u^2 + |x|^2 ) + c^4 ( v^2 + |y|^2 ) ) \\
&& +106 ( b^2 ( v^4 + |y|^4 ) + c^2 ( u^4 + |x|^4 ) ) \\
&& + 86 ( b^2 ( u^4 + |x|^4 ) + c^2 ( v^4 + |y|^4 ) ) \\
&& + 84 ( ( u^2 + v^2 ) ( |x|^4 + |y|^4 ) + ( |x|^2 + |y|^2 ) ( u^4 + v^4 ) ) \\
&& + 60 ( u^2 v^2 ( u^2 + v^2 ) + |x|^2 |y|^2 ( |x|^2 + |y|^2 ) ) \\
&& + 50 ( b^6 + c^6 ) 
+  20 ( u^6 + v^6 + |x|^6 + |y|^6 )
+   3 |bux + cvy|^2,
\end{eqnarray*}

\begin{eqnarray*}
q_7 &=& 802 b^2 c^2 ( u^2 |x|^2 + v^2 |y|^2 )
+778 b^2 c^2 ( u^2 |y|^2 + v^2 |x|^2 ) \\
&& +688 b^2 c^2 ( u^2 v^2 + |x|^2 |y|^2 ) \\
&& +574(b^2 v^2 |y|^2 (u^2 +|x|^2) 
    +c^2 u^2 |x|^2( v^2 + |y|^2 ) ) \\
&& +515 b^2 c^2 ( b^2 (v^2 + |y|^2) + c^2(u^2 +|x|^2) ) \\
&& +488( b^2 u^2 |x|^2 ( v^2 + |y|^2 ) 
        +c^2 v^2 |y|^2 ( u^2 + |x|^2 ) ) \\
&& +470 b^2 c^2 ( b^2 ( u^2 + |x|^2 ) + c^2 (v^2 +|y|^2) )
+ 418 ( b^4 v^2 |y|^2 + c^4 u^2 |x|^2 ) \\
&& +384 u^2 v^2 |x|^2 |y|^2 + 364 b^4 c^4
   +336 b^2 c^2 ( u^4 + v^4 + |x|^4 + |y|^4 ) \\
&& +331 ( b^4 + c^4 ) ( u^2 |y|^2 + v^2 |x|^2 ) \\
&& +324( b^2 v^2 |y|^2 ( v^2 + |y|^2 ) 
       + c^2 u^2 |x|^2 ( u^2 + |x|^2 ) ) \\
&& +278 ( b^4 + c^4 ) ( u^2 v^2 + |x|^2 |y|^2 ) \\
&& +274 ( u^2 |y|^2 ( b^2 |y|^2 + c^2 u^2 ) 
         +v^2 |x|^2 ( b^2 v^2 + c^2 |x|^2 ) ) \\
&& +260 ( b^4 u^2 |x|^2 + c^4  v^2 |y|^2 ) \\
&& +250 ( u^2 |y|^2 ( b^2 u^2 + c^2 |y|^2 ) 
        + v^2 |x|^2 ( b^2 |x|^2 + c^2 v^2 ) ) \\
&& +214 ( b^2 u^2 |x|^2 ( u^2 + |x|^2 ) 
        + c^2 v^2 |y|^2 ( v^2 + |y|^2 ) )
   +210 b^2 c^2 ( b^4 + c^4 ) \\
&& +204 ( u^2 v^2 ( b^2 v^2 + c^2 u^2 ) 
        +|x|^2 |y|^2 ( b^2 |y|^2 + c^2 |x|^2 ) ) \\
&& +192 ( u^2 v^2 ( |x|^4 + |y|^4 ) 
       +|x|^2 |y|^2 ( u^4 + v^4 ) ) \\
&& +180 ( u^2 v^2 ( b^2 u^2 + c^2 v^2 ) 
        + |x|^2 |y|^2 ( b^2 |x|^2 + c^2 |y|^2 ) ) \\
&& +174 ( b^4 ( v^4 + |y|^4 ) + c^4 ( u^4 + |x|^4 ) ) \\
&& +168(u^2 +v^2)(|x|^2 +|y|^2)(u^2 v^2 +|x|^2 |y|^2 ) \\
&& +135 ( b^6 ( v^2 + |y|^2 ) + c^6 ( u^2 + |x|^2 ) )
+ 112 ( b^4 ( u^4 + |x|^4 ) + c^4 ( v^4 + |y|^4 ) ) \\
&& +100 ( b^6 ( u^2 + |x|^2 ) + c^6 ( v^2 + |y|^2 ) )
+  96 ( u^4 + v^4 ) ( |x|^4 + |y|^4 ) \\
&& +76 ( b^2 ( v^6 + |y|^6 ) + c^2 ( u^6 + |x|^6 ) ) \\
&& +56 ( ( u^2 + v^2 ) ( |x|^6 + |y|^6 ) 
       + ( |x|^2 + |y|^2 ) ( u^6 + v^6 ) ) \\
&& +52 ( b^2 ( u^6 + |x|^6 ) + c^2 ( v^6 + |y|^6 ) )
+  48 ( u^4 v^4 + |x|^4 |y|^4 ) \\
&& +32 ( b^8 + c^8 + u^2 v^2 ( u^4 + v^4 ) 
                   + |x|^2 |y|^2 ( |x|^4 + |y|^4 ) ) \\
&& +( 10(b^2+c^2) + 7(u^2+v^2+|x|^2+|y|^2) )|bux + cvy|^2 \\
&& +8( u^8 + v^8 + |x|^8 + |y|^8 ),
\end{eqnarray*}

\begin{eqnarray*}
q_8 &=& 
1248 b^2c^2( u^2v^2(|x|^2+|y|^2) + |x|^2|y|^2(u^2 +v^2) ) \\
&& +960 b^2 c^2 ( b^2 v^2 |y|^2 + c^2 u^2 |x|^2 ) 
+820 b^2 c^2 ( b^2 + c^2 )( u^2 |y|^2 + v^2 |x|^2 ) \\
&& +780 b^2 c^2 ( b^2 u^2 |x|^2 + c^2 v^2 |y|^2 ) 
+776 u^2v^2 |x|^2 |y|^2 (b^2 +c^2 ) \\
&& +748 b^2 c^2 ( b^2 + c^2 ) ( u^2 v^2 + |x|^2 |y|^2 ) \\
&& +628 b^2 c^2 ( u^2 |x|^2 ( u^2 + |x|^2 ) 
                + v^2 |y|^2 ( v^2 + |y|^2 ) ) \\
&& +586 b^4 c^4 ( u^2 + v^2 + |x|^2 + |y|^2 ) \\
&& +576 b^2 c^2 ( u^2 |y|^2 ( u^2 + |y|^2 ) 
                + v^2 |x|^2 ( v^2 + |x|^2 ) ) \\
&& +570( b^4 v^2 |y|^2 ( u^2 + |x|^2 ) 
        +c^4 u^2 |x|^2 ( v^2 + |y|^2 ) ) \\
&& +488 b^2 c^2 ( u^2 v^2 ( u^2 + v^2 ) 
               + |x|^2 |y|^2 ( |x|^2 + |y|^2 ) ) \\
&& +470(b^2 v^2|y|^2 +c^2 u^2|x|^2)(u^2|y|^2+v^2|x|^2) \\
&& +428(b^2 v^2|y|^2 +c^2 u^2|x|^2)(u^2v^2 +|x|^2|y|^2) \\
&& +398b^2c^2( b^2( v^4 +|y|^4 ) +c^2( u^4 +|x|^4 ) ) \\
&& +394( b^4 u^2 |x|^2 ( v^2 + |y|^2 ) 
       + c^4 v^2 |y|^2 ( u^2 + |x|^2 ) ) \\
&& +382 ( b^2 v^2 |y|^2 ( u^4 + |x|^4 ) 
        + c^2 u^2 |x|^2 ( v^4 + |y|^4 ) ) \\
&& +366 ( b^4 v^2 |y|^2 ( v^2 + |y|^2 ) 
        + c^4 u^2 |x|^2 ( u^2 + |x|^2 ) ) \\
&& +358 b^2 c^2 ( b^4 ( v^2 + |y|^2 ) 
                + c^4 ( u^2 + |x|^2 ) ) \\
&& +332 b^2 c^2 ( b^2 ( u^4 + |x|^4 ) 
                + c^2 ( v^4 + |y|^4 ) ) \\
&& +320 ( b^2 u^2 |x|^2 ( v^4 + |y|^4 ) 
        + c^2 v^2 |y|^2 ( u^4 + |x|^4 ) ) \\
&& +306(b^2 u^2|x|^2 +c^2v^2|y|^2)(v^2|x|^2 +u^2|y|^2 ) \\
&& +304b^2 c^2( b^4(u^2 +|x|^2) + c^4(v^2 +|y|^2) ) \\
&& +284 ( b^2 v^4 |y|^4 + c^2 u^4 |x|^4 )
+ 276 b^4 c^4 ( b^2 + c^2 ) \\
&& +274(b^2u^2|x|^2 +c^2v^2|y|^2)(u^2v^2+|x|^2|y|^2 ) \\
&& +268 ( b^4 ( u^2 |y|^4 + v^4 |x|^2 ) 
        + c^4 ( u^4 |y|^2 + v^2 |x|^4 ) ) 
+256 ( b^6 v^2 |y|^2 + c^6 u^2 |x|^2 ) \\
&& +234 ( b^4 ( u^4 |y|^2 + v^2 |x|^4 ) 
        + c^4 ( u^2 |y|^4 + v^4 |x|^2 ) ) \\
&& +192 ( b^4 ( u^2 v^4 + |x|^2 |y|^4 ) 
        + c^4 ( u^4 v^2 + |x|^4 |y|^2 ) \\
&& \qquad + u^2 v^2|x|^2|y|^2(u^2 +v^2 +|x|^2 +|y|^2)) \\
&& +190 ( b^6 + c^6 ) ( u^2 |y|^2 + v^2 |x|^2 ) 
   +186 ( b^2 + c^2 ) ( u^4 |y|^4 + v^4 |x|^4 ) \\
&& +158 ( b^2 v^2 |y|^2 ( v^4 + |y|^4 ) 
        + c^2 u^2 |x|^2 ( u^4 + |x|^4 ) ) \\
&& +152 b^2 c^2 ( u^6 + v^6 + |x|^6 + |y|^6 ) \\
&& +140 ( b^4 u^2 |x|^2 ( u^2 + |x|^2 ) 
        + c^4 v^2 |y|^2 ( v^2 + |y|^2 ) 
        + ( b^6 + c^6 ) ( u^2 v^2 + |x|^2 |y|^2)) \\
&& +136 ( b^4 ( u^4 v^2 + |x|^4 |y|^2 ) 
        + c^4 ( u^2 v^4 + |x|^2 |y|^4 ) ) \\
&& +122 ( b^2 ( u^2 |y|^6 + v^6 |x|^2 ) 
        + c^2 ( u^6 |y|^2 + v^2 |x|^6 ) ) 
   +120 ( b^2 u^4 |x|^4 + c^2 v^4 |y|^4 ) \\
&& +112 ( b^2 |x|^2 ( b^4 u^2 + v^2 |x|^4 ) 
        + c^2 v^2 ( c^4 |y|^2 + v^4 |x|^2 )
        + u^2 |y|^2 ( b^2 u^4 + c^2 |y|^4 ) ) \\
&& +110 ( b^6 ( v^4 + |y|^4 ) + c^6 ( u^4 + |x|^4 ) )
+ 100  b^2 c^2 ( b^6 + c^6 ) \\
&& +96 ((b^2 +c^2)(u^4v^4+|x|^4|y|^4) 
        +(u^2+v^2)( |x|^4|y|^4 + u^2v^2(|x|^4+|y|^4 )) \\
&& \qquad +(|x|^2+|y|^2)(u^4v^4+|x|^2|y|^2( u^4 + v^4))) \\
\end{eqnarray*}

\begin{eqnarray*}
\quad\quad
&& +88 ( b^4 ( v^6 + |y|^6 ) + c^4 ( u^6 + |x|^6 ) ) \\
&& +80 ( b^2 ( u^2 v^6 + |x|^2 |y|^6 ) 
       + c^2 ( u^6 v^2 + |x|^6 |y|^2 ) ) \\
&& +76 ( b^2 u^2 |x|^2 ( u^4 + |x|^4 ) 
       + c^2 v^2 |y|^2 ( v^4 + |y|^4 ) ) \\
&& +64 ( u^2v^2(|x|^6+|y|^6+(|x|^2+|y|^2)(u^4+v^4)) \\
&& \qquad +|x|^2|y|^2(u^6+v^6+(u^2+v^2)(|x|^4+|y|^4)) ) \\
&& +52 ( b^8 ( v^2 + |y|^2 ) + c^8 ( u^2 + |x|^2 ) ) \\
&& +48 ( b^6 ( u^4 + |x|^4 ) + c^6 ( v^4 + |y|^4 ) 
       + u^2 v^2 ( b^2 u^4 + c^2 v^4 ) \\
&& \qquad +|x|^2 |y|^2 ( b^2 |x|^4 + c^2 |y|^4 ) ) \\
&& +32 ( ( u^6 + v^6 ) ( |x|^4 + |y|^4 )
       + ( |x|^6 + |y|^6 ) ( u^4 + v^4 ) \\
&& \qquad  + b^8 ( u^2 + |x|^2 ) + c^8 ( v^2 + |y|^2 ) 
    + b^4 ( u^6 + |x|^6 ) + c^4 ( v^6 + |y|^6 ) ) \\
&& +28 b^2 c^2 |bux + cvy|^2 
+  24 ( b^2 ( v^8 + |y|^8 ) + c^2 ( u^8 + |x|^8 ) ) \\
&& +18 ( b^2 ( v^2 + |y|^2 ) + c^2 ( u^2 + |x|^2 ) )
 |bux+cvy|^2 \\
&& +16 ( ( u^8 + v^8 ) ( |x|^2 + |y|^2 ) 
       + ( u^2 + v^2 ) ( |x|^8 + |y|^8 ) \\
&& \qquad +( u^2 + v^2 ) ( |x|^2 + |y|^2 ) |bux+cvy|^2 ) \\
&& +10( b^2(u^2+|x|^2) + c^2(v^2+|y|^2) ) |bux+cvy|^2 \\
&& +8( b^{10}+c^{10} + b^2(u^8+|x|^8) +c^2(v^8+|y|^8) + 
( b^4 + c^4 ) |bux+cvy|^2 ),
\end{eqnarray*}

\begin{eqnarray*}
q_9 &=& 492 b^2 c^2 u^2 v^2 |x|^2 |y|^2 \\
&& +349 b^2 c^2 ( b^2 v^2 |y|^2 ( u^2 + |x|^2 ) 
                + c^2 u^2 |x|^2 ( v^2 + |y|^2 ) ) \\
&& +305 b^2 c^2 ( b^2 u^2 |x|^2 ( v^2 + |y|^2 ) 
                + c^2 v^2 |y|^2 ( u^2 + |x|^2 ) ) \\
&& +260 b^4 c^4 ( u^2 |x|^2 + v^2 |y|^2 ) \\
&& +231 b^2 c^2 ( u^2 |x|^2 + v^2 |y|^2 ) ( u^2 v^2 + |x|^2 |y|^2 ) \\
&& +230 b^4 c^4 ( u^2 + |x|^2 ) ( v^2 + |y|^2 ) \\
&& +228 b^2 c^2 ( u^2 v^2 ( |x|^4 + |y|^4 ) 
               + |x|^2 |y|^2 ( u^4 + v^4 ) ) \\
&& +216 b^2 c^2 ( u^2 |x|^2 ( v^4 + |y|^4 ) 
                + v^2 |y|^2 ( u^4 + |x|^4 ) ) \\
&& +200 b^2 c^2 ( b^2 v^2 |y|^2 ( v^2 + |y|^2 ) 
                + c^2 u^2 |x|^2 ( u^2 + |x|^2 ) ) \\
&& +164 b^2 c^2 ( b^4 v^2 |y|^2 + c^4 u^2 |x|^2 ) \\
&& +149 b^2 c^2 ( b^2 ( u^2 |y|^4 +  v^4 |x|^2 ) 
                + c^2 ( u^4 |y|^2 +  v^2 |x|^4 ) ) \\
&& +146 u^2 v^2 |x|^2 |y|^2 ( b^4 + c^4 ) \\
&& +144 b^2 c^2 ( b^2 ( u^4 |y|^2 + u^2 v^4 + v^2 |x|^4 + |x|^2 |y|^4 ) \\
&& \qquad\quad + c^2(u^2|y|^4+u^4 v^2 + v^4 |x|^2 + |x|^4 |y|^2 ) ) \\
&& +140 b^2 c^2 ( b^2 u^2 |x|^2 ( u^2 + |x|^2 ) 
                + c^2 v^2 |y|^2 ( v^2 + |y|^2 ) ) \\
&& +131 b^2 c^2 ( b^4 + c^4 )( u^2 |y|^2 + v^2 |x|^2 ) \\
&& +125 b^2 c^2 ( b^2 ( u^4 v^2 + |x|^4 |y|^2 ) 
                + c^2 ( u^2 v^4 + |x|^2 |y|^4 ) ) \\
&& +120 (  b^2 c^2 ( u^4 |x|^4 +  v^4 |y|^4 ) 
        + ( b^4 v^2 |y|^2  + c^4 u^2 |x|^2 ) ( u^2 |y|^2 +  v^2 |x|^2 ) ) \\
&& +117 b^2 c^2( b^4 + c^4 )( u^2 v^2 + |x|^2 |y|^2 ) \\
&& +112 ( b^2 c^2 ( b^4 u^2 |x|^2 + c^4 v^2 |y|^2 ) \\
&& \qquad + u^2 v^2 |x|^2 |y|^2 ( b^2 ( v^2 + |y|^2 ) + c^2 ( u^2 + |x|^2 ) ) ) \\
&& +111 ( b^4 v^2 |y|^2 + c^4 u^2 |x|^2 ) ( u^2 v^2 +  |x|^2 |y|^2 ) \\
&& +110 b^4 c^4 ( b^2 ( v^2 + |y|^2 ) 
                + c^2 ( u^2 + |x|^2 ) ) \\
&& +108 b^4 c^4 ( u^4 + v^4 + |x|^4 + |y|^4 ) \\
&& +104 b^4 c^4 ( b^2 ( u^2 + |x|^2 ) 
                + c^2 ( v^2 + |y|^2 ) ) \\
&& +93 ( b^4 v^2 |y|^2 ( u^4 + |x|^4 ) 
        +c^4 u^2 |x|^2 ( v^4 + |y|^4 ) ) \\
&& +92 ( b^4 v^4 |y|^4 + c^4 u^4 |x|^4 ) 
 +90 b^2 c^2 ( u^4 |y|^4 +  v^4 |x|^4 ) \\
&& +80 ( b^2 c^2 ( u^4 v^4 + |x|^4 |y|^4 ) 
       + b^2 v^4 |y|^4 ( u^2 + |x|^2 ) 
       + c^2 u^4 |x|^4 ( v^2 + |y|^2 ) \\
&& \qquad + u^2 v^2|x|^2|y|^2(b^2 ( u^2 + |x|^2 ) 
                           + c^2 ( v^2 + |y|^2 ) ) ) \\
&& +76 b^2 c^2 ( u^2 |x|^2 ( u^4 + |x|^4 ) 
               + v^2 |y|^2 ( v^4 + |y|^4 ) ) \\
&& +74 ( b^6 v^2 |y|^2 ( u^2 + |x|^2 ) 
       + c^6 u^2 |x|^2 ( v^2 + |y|^2 ) ) \\
&& +72 ( u^4 v^4 (  b^2 |y|^2 + c^2 |x|^2 ) 
      + |x|^4 |y|^4 (  b^2 v^2 + c^2 u^2 ) ) \\
&& +70 b^2 c^2 ( b^4( v^4 + |y|^4) + c^4( u^4 +|x|^4)) \\
&& +64 ( u^4 |y|^4 (  b^2 v^2 + c^2 |x|^2 ) 
       + v^4 |x|^4 (  b^2 |y|^2 + c^2 u^2 ) ) \\
&& +61 b^2 c^2 ( u^2 |y|^2 ( u^4 + |y|^4 ) 
               + v^2 |x|^2 ( v^4 + |x|^4 ) ) \\
&& +57 ( b^4 u^2 |x|^2 ( v^4 + |y|^4 ) 
       + c^4 v^2 |y|^2 ( u^4 + |x|^4 ) ) \\
&& +56(b^6 c^6+ b^2 v^2 |y|^2 ( u^2 v^4 + |x|^2 |y|^4 ) 
              + c^2 u^2 |x|^2 ( u^4 v^2 + |x|^4 |y|^2)) \\
&& +54 ( b^6 v^2 |y|^2 ( v^2 + |y|^2 ) 
       + c^6 u^2 |x|^2 ( u^2 + |x|^2 ) ) \\
\end{eqnarray*}

\begin{eqnarray*}
\quad && +48 ( b^2 c^2 ( u^2 v^2 ( u^4 + v^4 ) 
                +|x|^2|y|^2(|x|^4+|y|^4) \\
&& \qquad\quad + b^4 ( u^4 + |x|^4 ) + c^4 ( v^4 + |y|^4 ) ) 
               + b^2 ( v^6 + |y|^6 ) + c^2 ( u^6 + |x|^6 ) \\
&& \qquad + b^2 v^2 |y|^2 ( u^2 |y|^4 + v^4 |x|^2 )
         + c^2 u^2 |x|^2 ( u^4 |y|^2 + v^2 |x|^4  ) ) \\
&& +46 ( b^4 v^2 |y|^2 ( v^4 + |y|^4 ) 
       + c^4 u^2 |x|^2 ( u^4 + |x|^4 ) ) \\
&& +45 ( b^4 + c^4 ) ( u^4 |y|^4 + v^4 |x|^4 ) \\
&& +44 (b^4 u^2|x|^2+c^4 v^2|y|^2)(u^2|y|^2+v^2|x|^2) \\
&& +40 ( b^2 v^2 |y|^2 ( u^6 + |x|^6 ) + c^2 u^2 |x|^2 ( v^6 + |y|^6 ) \\
&& \qquad +b^2 c^2 ( b^6 ( v^2 + |y|^2 ) 
                   + c^6 ( u^2 + |x|^2 ) ) ) \\
&& +39 ( b^6 ( u^2 |y|^4 + v^4 |x|^2 ) 
       + c^6 ( v^2 |x|^4 + u^4 |y|^2 ) ) 
+36 b^4 c^4 ( b^4 + c^4 ) \\
&& +34 ( b^6 ( u^4 |y|^2 + v^2 |x|^4 ) 
       + c^6 ( v^4 |x|^2 + u^2 |y|^4 ) ) \\
&& +33 ( b^4 ( v^6 |x|^2 + u^2 |y|^6 )        
       + c^4 ( u^6 |y|^2 + v^2 |x|^6 ) ) \\
&& +32 ( b^2 v^4 |y|^4 ( v^2 + |y|^2 ) 
       + c^2 u^4 |x|^4 ( u^2 + |x|^2 ) \\
&& \qquad + b^4 c^2 ( u^6 + |x|^6 ) + b^2 c^4 ( v^6 + |y|^6 ) 
        + b^8 c^2 ( u^2 + |x|^2 ) \\
&& \qquad + b^2 c^8 ( v^2 + |y|^2 ) 
          + b^6 u^2 |x|^2 ( v^2 + |y|^2 ) 
          + c^6 v^2 |y|^2 ( u^2 + |x|^2 ) \\
&& \qquad + b^2 u^2 |x|^2 ( u^2 |y|^4 + v^4 |x|^2 ) 
         + c^2 v^2 |y|^2 ( v^2 |x|^4 + u^4 |y|^2 ) ) \\
&& +28 ( b^4 ( u^6 |y|^2+ v^2 |x|^6 ) 
       + c^4 ( v^6 |x|^2 + u^2 |y|^6 ) ) \\
&& +24 ( b^2 u^2 |x|^2 ( v^6 + |y|^6 ) 
       + c^2 v^2 |y|^2 ( u^6 + |x|^6 ) \\
&& \qquad + (b^4 u^2 |x|^2 + c^4 v^2 |y|^2 ) ( u^2 v^2 + |x|^2 |y|^2 ) \\
&& \qquad + u^4 v^4 (  b^2 |x|^2 + c^2 |y|^2 ) 
         + |x|^4 |y|^4 (  b^2 u^2 + c^2 v^2 ) ) \\
&& +20 ( b^8 v^2 |y|^2 + c^8 u^2 |x|^2 ) 
 +17 ( b^6 ( u^2 v^4 + |x|^2 |y|^4 ) 
       + c^6 ( u^4 v^2 + |x|^4 |y|^2 ) ) \\
&& +16 ( ( b^2 + c^2 ) ( u^4 |y|^4  ( u^2 + |y|^2 ) 
                       + v^4 |x|^4  ( v^2 + |x|^2 )) \\
&& \qquad + ( b^8 + c^8 ) ( u^2 |y|^2 + v^2 |x|^2 )
          + u^2 v^2 ( b^2 |x|^6 + c^2 |y|^6 ) \\
&& \qquad  + |x|^2 |y|^2 ( b^2 u^6 + c^2 v^6 ) 
          + b^2 v^2 |y|^2 ( v^6 + |y|^6 ) 
          + c^2 u^2 |x|^2 ( u^6 + |x|^6 ) \\
&& \qquad + b^2 u^4 |x|^4 ( v^2 + |y|^2 ) 
          + c^2 v^4 |y|^4 ( u^2 + |x|^2 ) \\
&& \qquad  + b^4 ( u^2 v^6 + |x|^2 |y|^6 ) 
          + c^4 ( u^6 v^2 + |x|^6 |y|^2 ) \\
&& \qquad  + b^6 ( v^6 + |y|^6 ) + c^6 ( u^6 + |x|^6 )
          + (b^2v^2|y|^2 +c^2u^2|x|^2) |bux+cvy|^2 ) \\
&& +10 ( b^8 ( v^4 + |y|^4 ) + c^8 ( u^4 + |x|^4 ) \\
&& \qquad  + b^2c^2(u^2+v^2+|x|^2+|y|^2) |bux+cvy|^2 ) \\
&& +8 ( b^{10} c^2 + b^2 c^{10} + b^4 ( v^8 + |y|^8 ) 
                            + c^4 ( u^8 + |x|^8 ) \\
&& \qquad + ( b^4 + c^4 ) ( u^4 v^4 + |x|^4 |y|^4 ) 
         + v^2 |x|^2 ( b^2 + c^2 ) ( v^6 + |x|^6 ) \\
&& \qquad + u^2 |y|^2 ( b^2 + c^2 ) ( u^6 + |y|^6 )
         + b^2 c^2 ( |x|^8 + |y|^8 + u^8 + v^8 ) \\
&& \qquad + b^2 u^2 |x|^2 ( u^4 v^2 + |x|^4 |y|^2 )
         + c^2 v^2 |y|^2 ( |x|^2 |y|^4 + u^2 v^4 ) \\
&& \qquad + ( b^2 + c^2 ) ( b^2 c^2 + u^2 |y|^2 + v^2 |x|^2 ) |bux+cvy|^2 ) \\
&& +4 ( ( b^8 + c^8 ) ( u^2 v^2 + |x|^2 |y|^2 ) 
        + ( b^4 ( v^2 + |y|^2 ) + c^4 ( u^2 + |x|^2 ) )
|bux+cvy|^2 ) \\
&& +2 ( b^{10} ( v^2 + |y|^2 ) + c^{10} ( u^2 + |x|^2 ) \\
&& \qquad  + b^6 ( u^4 v^2 + |x|^4 |y|^2 )
        + c^6 ( u^2 v^4 + |x|^2 |y|^4 ) ).
\end{eqnarray*}

\end{document}